\documentclass[hidelinks,onefignum,onetabnum]{siamart220329}

\usepackage{graphicx}
\usepackage{amssymb}

\usepackage{lipsum} 

\newcommand{\rnd}{\mathrm{rnd}}

\newcommand{\cb}{\color{black}}

\usepackage{soul}

\global\long\def\R{\mathbb{R}}%
\global\long\def\L{\mathcal{L}}%
\global\long\def\K{\mathcal{K}}%
\global\long\def\D{\mathcal{D}}%
\global\long\def\sign{\textrm{sign}}%

\newcommand{\bi}{\begin{itemize}}

\newcommand{\ei}{\end{itemize}}
\newcommand{\bd}{\begin{description}}
\newcommand{\ed}{\end{description}}
\newcommand{\be}{\begin{enumerate}}
\newcommand{\ee}{\end{enumerate}}
\newtheorem{thm}{Theorem}
\newtheorem{defn}[thm]{Definition}

\newcommand{\bqn}{\begin{eqnarray}}
\newcommand{\eqn}{\end{eqnarray}}
\newcommand{\eqnn}{\nonumber\end{eqnarray}}
\newcommand{\eqnl}[1]{\label{#1}\end{eqnarray}}
\newcommand{\bga}{\begin{gather}}
\newcommand{\ega}{\end{gather}}

\newcommand{\ba}[1]{\begin{array}{#1}}
\newcommand{\ea}{\end{array}}

\usepackage{xcolor}
\newcommand{\beq}{\begin{equation}}
\newcommand{\eeq}{\end{equation}}

\newcommand{\bag}{\begin{aligned}}
\newcommand{\eag}{\end{aligned}}

\usepackage[english]{babel}

 \theoremstyle{definition}
 \theoremstyle{plain}
 \theoremstyle{plain}

 \makeatother

\usepackage{tikz}
\usepackage{pgfplots}
\usepackage{subfigure}
\usepackage{tikz-3dplot}
\usetikzlibrary{calc,arrows,shapes,backgrounds,calc,positioning,patterns,decorations.pathmorphing,decorations.markings,mindmap,trees,fadings,angles,quotes}
\usepackage{blox}
\usepackage{color}
\usepackage{booktabs}
\usepackage{balance}
\usepackage{xcolor}

\usepackage{dsfont}

\usepackage{amsmath} 
\usepackage{graphicx} 
\usepackage{caption} 

\usepackage{lipsum}
\usepackage{amsfonts}
\usepackage{graphicx}
\usepackage{epstopdf}
\usepackage{algorithmic}
\ifpdf
  \DeclareGraphicsExtensions{.eps,.pdf,.png,.jpg}
\else
  \DeclareGraphicsExtensions{.eps}
\fi

\newsiamremark{remark}{Remark}
\newsiamremark{hypothesis}{Hypothesis}
\crefname{hypothesis}{Hypothesis}{Hypotheses}
\newsiamthm{claim}{Claim}

\headers{Robustness of Delayed Higher Order Sliding Mode Control}{Moussa Labbadi, Denis Efimov, and Leonid Fridman}

\title{Robustness of Delayed Higher Order Sliding Mode Control}

\author{Moussa Labbadi\thanks{Aix-Marseille University, LIS UMR CNRS 7020, 13013 Marseille, France
  (\email{moussa.labbadi@lis-lab.fr}.}
\and Denis Efimov \thanks{Inria, Univ. Lille, CNRS, UMR 9189 - CRIStAL, F-59000 Lille, France
  (\email{denis.efimov@inria.fr}.}
\and Leonid Fridman\thanks{ Facultad de Ingenieria, Universidad Nacional Autonoma de Mexico, Mexico City 04510, Mexico, and Department of Information Technologies and AI, Sirius University of Science and Technology, 354340 Sochi, Russia
  (\email{lfridman@unam.mx}.}}

\usepackage{amsopn}
\DeclareMathOperator{\diag}{diag}

\ifpdf
\hypersetup{
  pdftitle={Robustness of Delayed Higher Order Sliding Mode Control},
  pdfauthor={Moussa Labbadi, Denis Efimov, and Leonid Fridman}
}

\fi

\begin{document}

\maketitle

\begin{abstract}
In this paper,  the feasibility of recently developed higher order delayed sliding mode controllers is addressed. With this aim the robustness against the measurement noise and mismatched perturbations for the systems governed by such controllers is established using ISS implicit Lyapunov-Razumikhin function approach. To illustrate proposed results, a simulation example validating the efficiency of the method is provided.

\end{abstract}



\section{Introduction}

{\cb Sliding mode control (SMC) is one of the most popular techniques demonstrating two main advantages \cite{utkin1977variable,shtessel2014sliding}: theoretically exact compensation of matched perturbations and finite-time convergence to the sliding set (model reduction in a finite time). 
Further, the Higher Order SMC (HOSMC) algorithms \cite{shtessel2014sliding,lee2007chattering,levant1993sliding,hosm:Levant:2003,levant2005homogeneity,Levant2005b,bartolini2003survey,cruz2019higher,mercado2020discontinuous,sanchez2019design} ensuring the same quality properties for systems with arbitrary relative degree have been designed.  Unfortunately, the price to pay for the mentioned advantages of SMC is the usage of discontinuous controllers having infinite control gains, whose application requires an arbitrary high frequency of switching for actuators, which causes three key issues for successful implementation of SMC: $(i)$ chattering investigation and adjustment, where the main directions for chattering mitigation include adaptation of the SMC gains and utilization of observers to bypass infinite gains \cite{lee2007chattering,levant1993sliding,hosm:Levant:2003,levant2005homogeneity,Levant2005b,bartolini2003survey,utkin2017sliding,pilloni2012parameter,perez2019design,boiko2014relative}; $(ii)$ minimization of effects of unmatched perturbations, where the main directions are reasonable choice of sliding surfaces \cite{Choi1999,Castanos2006}, backstepping based SMC \cite{Estrada2017,Zhang2020} or observer-based compensation \cite{FerreiradeLoza2015}; $(iii)$ robustness with respect to the measurement noise, which is traditionally studied using homogeneity of SMC \cite{levant1993sliding,hosm:Levant:2003,levant2005homogeneity,Levant2005b,bartolini2003survey,mercado2020discontinuous,sanchez2019design}.

Recently, a new control law, so called delayed SMC algorithm \cite{efimov2016delayed}, was proposed that generates a continuous control signal for all initial functions different from zero and ensuring theoretically exact compensation of matching uncertainties (similarly to the conventional SMCs), but providing only hyperexponential convergence to the first order sliding set (the stability analysis in time-delay systems is considerably more complex than in ordinary differential equations due to the challenging task of designing Lyapunov-Krasovskii functionals or Lyapunov-Razumikhin functions \cite{gu2003introduction,kolmanovskii1986stability}). In the paper \cite{Labbadi2024} such algorithms were generalized to the systems with arbitrary relative degrees with corresponding attenuation of chattering using the Implicit Lyapunov function (ILF) approach \cite{polyakov2015finite,Polyakov2016} (for the noise-free setting). That is why the next step in development of the delayed HOSMCs is investigation of their robustness with respect to unmatched perturbations and measurement noises, which is the goal of this work.}

To illustrate this discussion and to motivate our paper, consider a disturbed double integrator system: 
\begin{equation}
\dot{x}_{1}(t)  =x_{2}(t),\;\
\dot{x}_{2}(t)  =u(t)+d(t),\;\
y(t)  =x(t)+w(t),
\label{eq:system}
\end{equation}
where $x(t)=[x_{1}(t),x_{2}(t)]^{\top}\in\mathbb{R}^{2}$ represents the state vector, $y(t)\in\R^{2}$ is the measured output, $u(t)\in\mathbb{R}$ is the control input, $d(t)\in\mathbb{R}$ denotes a bounded disturbance,
and $w(t)\in\mathbb{R}^{2}$ is a bounded noise signal.  First-order sliding mode controllers \cite{utkin1977variable}, $u(t)=-k\sign(\sigma(t))$ with $k>0$ sufficiently big, can exponentially stabilize the origin of the system \eqref{eq:system} (in the noise-free setting) by confining the system dynamics to a desired sliding surface $\sigma(t)=y_{2}(t)+ay_{1}(t)$, $a>0$ in a finite time. However, the chattering issue, marked by high-frequency oscillations of the control signal with the amplitude
$\pm k$ when trajectories linger near the sliding surface, poses
risks of actuator damage and transient performance degradation
\cite{utkin1977variable,shtessel2014sliding}. To this end, various
methods have been devised to mitigate chattering \cite{lee2007chattering}.
High-order SMCs (HOSMCs), as widely applied super-twisting control law $u(t)=-\ell_{1}\sqrt{|\sigma(t)|}\sign(\sigma(t))+z(t)$,
$\dot{z}(t)=-\ell_{2}\sign(\sigma(t))$ with $\ell_{1},\ell_{2}>0$ is applicable for bounded and sufficiently smooth matched disturbances, representing one of the most prominent approaches for chattering reduction. However, implementation of existing HOSMC algorithms is challenging due to the limited availability of constructive procedures for tuning of control parameters and discrete-time realization, especially beyond the second-order SMC \cite{pilloni2012parameter,perez2019design}.

For constructive design in nonlinear cases, the ILF approach \cite{MCloskey1997,adamy2004soft} is frequently employed (the controllability function method was introduced in \cite{korobov1979general} for control design). The noise and unmatched disturbance robustness analysis (in the input-to-state stability (ISS) sense \cite{sontag2007input}) is usually performed using the robust features of stable homogeneous systems \cite{Bernuau2013,Levant2016}. Application of ILF-based control laws requires on-line computation of the value $V_{y}(y(t))$ of this function \cite{polyakov2015finite}:
$u(t)=-k_{1}\frac{y_{1}(t)}{V_{y}^{2}(y(t))}-k_{2}\frac{y_{2}(t)}{V_{y}(y(t))}$,
where $p_{11}y_{1}^{2}(t)+2p_{12}V_{y}(y(t))y_{1}(t)y_{2}(t)+p_{22}V_{y}^{2}(y(t))y_{2}^{2}(t)=V_{y}^{4}(y(t))$,
and $k_{1},k_{2},p_{11},p_{22}>0$ with $p_{12}\in\R$ are suitably
tuned for \eqref{eq:system}. Such a control generates chattering
at the origin as the quasi-continuous SMC analogues \cite{Levant2005b}. 

For the \emph{measurement noise-free case}, in \cite{Labbadi2024} the ILF-based HOSMC from \cite{polyakov2015finite,Polyakov2016} was developed for \eqref{eq:system}
\begin{equation}
\begin{aligned}u(t)=-k_{1}\frac{y_{1}(t)}{\Psi^{2}(y_{t})}-k_{2}\frac{y_{2}(t)}{\Psi(y_{t})},\end{aligned}
\label{eq:u2}
\end{equation}
where 
\begin{gather}
\Psi(y_{t})=\max\left[V_{y}\left(y(t)\right),e^{1-\chi}\min\left[\underset{-\eta\leq\theta\leq0}{\max}V_{y}\left(y(t+\theta)\right),\right.\right.\nonumber \\
\left.\left.\left(\underset{-\eta\leq\theta\leq0}{\max}V_{y}\left(y(t+\theta)\right)\right)^{\chi}\right]\right]\label{eq:Psi}
\end{gather}
with $\chi>1$ and $\eta>0$ being tuning parameters (the ILF $V_y(y(t))$ is calculated as above). {\cb For $w\equiv0$,} this control provides uniform (independently in a properly bounded matched disturbance $d$) hyperexponential (faster than any exponential) convergence to the origin, being uniformly upper
bounded as other HOSMC. The tuning of control parameters and stability conditions are expressed using Linear Matrix Inequalities (LMIs), facilitating straightforward adjustments for convergence times and for magnitudes of counteracted disturbances.
{\cb In (\ref{eq:u2}), (\ref{eq:Psi}) the chattering happens if $\Psi(y_t)=0$, which corresponds to the situation when $y(t+\theta)=0$ for all $\theta\in[-\eta,0]$ only, while the origin is reached asymptotically with a hyperexponential convergence rate. Therefore, the chattering has a much lower chance to be produced. Moreover,} the behavior of trajectories with a hyperexponential rate of convergence is difficult to distinguish in simulations/experiments with the finite-time decaying ones. Since \eqref{eq:u2} generates a continuous signal and the trajectories of the system \eqref{eq:system}, \eqref{eq:u2} go to the origin, it implies that \eqref{eq:u2} has to approach $-d(t)$ or its average, i.e., performing identification of the matched uncertainty.

In the present note, the robustness against measurement noises $w(t)$ {\cb and additional mismatched perturbations}
is analyzed for the control method of \cite{Labbadi2024}. {\cb The technical difficulty consists in impossibility of exploitation of the properties of homogeneous systems, as it is usually performed \cite{polyakov2015finite,Polyakov2016} since the system is not homogeneous due to $\chi>1$, and in addition,
it contains delays.} The Lyapunov-Razumikhin function approach will be applied.


\section*{Notation}
\begin{itemize}
\item The set of real numbers is denoted by $\mathbb{R}$, then $\mathbb{R}_{+}=\{x\in\mathbb{R}:x\geq0\}$
and $\mathbb{R}_{+}^{\star}=\R_{+}\setminus\{0\}$ are the sets of
nonnegative and positive reals, respectively. The Euclidean norm of
a vector $x\in\R^{n}$ is denoted by $|x|$. 
\item A diagonal matrix with elements $\nu_{i}$, $i=1,\dots,n$ on the
main diagonal is denoted by $\diag\{\nu_{i}\}_{i=1}^{n}$. 
\item For a symmetric matrix $P\in\R^{n\times n}$, the minimum and maximum
eigenvalues are represented by $\lambda_{\text{min}}(P)$ and $\lambda_{\text{max}}(P)$. 
\item Denote the identity matrix of dimension $n\times n$ by $I_{n}$.
\item For two metric spaces $X$ and $Y$, the set of continuous maps between
them is denoted by $C(X,Y)$. The Banach space of continuous functions
$C([-\tau,0],\R^{n})$ with a finite $\tau>0$ will be denoted by
$C_{\tau}$ and equipped with the uniform norm $\|\varphi\|=\underset{-\tau\leq\xi\leq0}{\max}|\varphi(\xi)|$
for $\varphi\in C_{\tau}$. 
\item For a (Lebesgue) measurable function $d:\mathbb{R}_{+}\to\mathbb{R}^{s}$
and $[t_{0},t_{1})\subset\R_{+}$ define the norm $\|d\|_{[t_{0},t_{1})}=\text{ess\ sup}_{t\in[t_{0},t_{1})}\Vert d(t)\Vert$,
then $\|d\|_{\infty}=\|d\|_{[0,+\infty)}$ and the set of $d$ with
the property $\|d\|_{\infty}<+\infty$ we denote as $\mathcal{L}_{\infty}^{s}$
(i.e., this is the set of essentially bounded measurable functions). 
\item For a locally Lipschitz continuous function $V:\mathbb{R}^{n}\to\mathbb{R}_{+}$,
the upper directional Dini derivative is defined as follows: 
\(
D^{+}V(x)v=\limsup\limits _{h\to0^{+}}\frac{V\left(x+hv\right)-V\left(x\right)}{h}
\) for any $x\in\R^{n}$ and $v\in\R^{n}$.
\item A function $\sigma\in C(\mathbb{R}_{+},\mathbb{R}_{+})$ belongs to
class $\mathcal{K}$ if it is strictly increasing and $\sigma(0)=0$;
it additionally belongs to class $\mathcal{K}_{\infty}$ if it is
also unbounded. A function $\beta\in C(\mathbb{R}_{+}\times\mathbb{R}_{+},\mathbb{R}_{+})$
belongs to class $\mathcal{KL}$ if $\beta(\cdot,r)\in\mathcal{K}$
for any $r\in\R_{+}$ and $\beta(r,\cdot)$ is decreasing to zero
for any $r>0$.
\item A function $\hbar\in C(\mathbb{R}_{+}^{\star}\times\mathbb{R}_{+}^{\star},\mathbb{R})$
is said to be of the class $\mathcal{IK}_{\infty}$ (implicit $\K_{\infty}$)
\cite{polyakov2015implicit} if: 
1) for any $s\in\mathbb{R}_{+}^{\star}$ there exists $\alpha\in\mathbb{R}_{+}^{\star}$
such that $\hbar(\alpha,s)=0$; 
2) for any fixed $s\in\mathbb{R}_{+}^{\star}$, the function $\hbar(\cdot,s)$
is strictly decreasing on $\mathbb{R}_{+}^{\star}$; 
3) for any fixed $\alpha\in\mathbb{R}_{+}^{\star}$, the function $\hbar(\alpha,\cdot)$
is strictly increasing on $\mathbb{R}_{+}^{\star}$; 
4) for all $(\alpha,s)\in\Gamma=\{(\alpha,s)\in\mathbb{R}_{+}^{\star}\times\mathbb{R}_{+}^{\star}:\hbar(\alpha,s)=0\}$:
\(
\lim_{s\to0^{+}}\alpha=0,\quad\lim_{\alpha\to0^{+}}s=0,\quad\lim_{s\to+\infty}\alpha=+\infty.
\)
\end{itemize}

\section{Problem Statement}

\label{sec:II}

This paper addresses the problem of uniform stabilization at the origin
with an accelerated (hyperexponential) rate of convergence for a linear
single-input system with matched {\cb and mismatched} perturbations of the form: 
\begin{align}
\dot{x}(t) & =Ax(t)+b(u(t)+d(t)){\color{black}+\delta(t)},\quad t\geq0,\label{eq:LS}\\
y(t) & =x(t)+w(t), 
\end{align}
where 
\[
A=\begin{bmatrix} 0 & 1 & 0 & \cdots & 0\\
0 & 0 & 1 & \cdots & 0\\
\vdots & \vdots & \vdots & \ddots & \vdots\\
0 & 0 & 0 & \cdots & 1\\
0 & 0 & 0 & \cdots & 0
\end{bmatrix},\quad b=\begin{bmatrix} 0\\
0\\
\vdots\\
0\\
1
\end{bmatrix},
\]
$x(t)\in\mathbb{R}^{n}$ is the state vector, $u(t)\in\mathbb{R}$
is the control input; $d\in\L^1_{\infty}$ represents the matched disturbances,
where $\|d\|_{\infty}\leq\Delta$ for a given $\Delta>0$ (further,
we denote by $\mathcal{D}$ the set of all inputs $d\in\L^1_{\infty}$
with $\|d\|_{\infty}<\Delta$, $D=\{d\in\R:|d|\leq\Delta\}$); {\color{black} the mismatched disturbance $\delta\in\L_{\infty}^n$ is acting on the main dynamics, and we will assume that $b^\top\delta(t)=0$ for all $t\geq0$, i.e., the matched component of $\delta$ is included in $d$;
}
the entire state vector $x(t)$ is assumed to be measured by the output
$y(t)\in\R^{n}$ subject to the noise $w\in\L_{\infty}^{n}$.  Recall that many controllable systems can be presented in the form (\ref{eq:LS}) also including the remaining terms in the disturbances {\cb $d$ and $\delta$ (which may be dependent on $x$ provided that their uniform boundedness is kept).} 

{\color{black}It is required to find a control law that globally stabilizes \eqref{eq:LS}
at the origin uniformly in any disturbance $d\in\mathcal{D}$ while
$\|w\|_{\infty}=\|\delta\|_{\infty}=0$, and achieving ISS property for $\delta\in\L_{\infty}^n$ and $w\in\mathcal{L}_{\infty}^{n}$,
providing an accelerated rate of convergence. 
Nevertheless, we allow the control function $u$ to
be dependent on the past values of the output $y(t)$ being discontinuous with respect to $x$, hence, this paper utilizes
the theory of \cite{Kolmanovskii1999} to define the solutions.}

\section{Preliminaries}

\label{sec:III}

In this section, first, the HOSMC design based on ILF is recalled,
next for time-delay systems several definitions of stability with
accelerated convergence rates are given, together with the Lyapunov-Razumikhin
conditions establishing these properties. Finally, the ISS definition
with related Lyapunov-Razumikhin result are formulated.

Following the problem statement, we look for the control design ensuring
an accelerated convergence in the noise-free case and the ISS property
in the presence of a bounded noise $w$ {\cb and state perturbation $\delta$}, uniformly with respect to
the matched disturbance $d$, while the definitions and the Lyapunov-Razumikhin
conditions presented below consider these properties separately, as
they usually given in the literature. In the paper we will straightforwardly
combine them to get the required result. 

\subsection{ILF-based HOSMC}

The regulators obtained in \cite{polyakov2015finite,Polyakov2016}
via the ILF approach can be chosen based on the following conditions\footnote{The definitions of conventional stability concepts for ordinary differential
equations, as \eqref{eq:LS}, can be found in \cite{Khalil2002,Efimov2021}.}: 
\begin{thm}
\label{thm:HOSM_ctrl} \cite{Polyakov2016} Take the control in the
form 
\begin{equation}
u(y)=YX^{-1}D_{\mathbf{r}}\left(V_{y}^{-1}(y)\right)y\label{eq:HOSM_ctrl}
\end{equation}
with $D_{\mathbf{r}}\left(\lambda\right)=\diag\{\lambda^{r_{i}}\}_{i=1}^{n}$
for $\mathbf{r}=(r_{1},\dots,r_{n})=(n,...,1)$ and $X\in\R^{n\times n}$,
$Y\in\R^{1\times n}$ solving the following LMIs 
\begin{gather}
\left\{ \begin{array}{l}
X\succ0,\;\left[\begin{array}{cc}
\Upsilon & b\\
b^{\top} & -\varrho_{2}
\end{array}\right]\preceq0,\\
\left[\begin{array}{cc}
XG_{\mathbf{r}}+G_{\mathbf{r}}X & X\\
X & I_{n}
\end{array}\right]\succeq0,
\end{array}\right.\label{eq:LMI_imp_Rn}\\
\Upsilon=XA^{\top}+AX+Y^{\top}b^{\top}+bY+\varrho_{1}(XG_{\mathbf{r}}+G_{\mathbf{r}}X),\nonumber 
\end{gather}
where $\varrho_{1}>\varrho_{2}\Delta^{2}>0$ and $G_{\mathbf{r}}=\text{diag}\{r_{i}\}_{i=1}^{n}$,
while the ILF $V_{y}:\R^{n}\to\R_{+}$ is defined as 
\begin{gather}
Q(V_{y}(y),y)=0,\nonumber \\
Q(V,y)=y^{\top}D_{\mathbf{r}}\left(V^{-1}\right)X^{-1}D_{\mathbf{r}}\left(V^{-1}\right)y-1\label{eq:Q}
\end{gather}
for $y\ne0$ and $V_{y}(0)=0$ otherwise. Then on all trajectories
of the system \eqref{eq:LS}, \eqref{eq:HOSM_ctrl} with $w\equiv0$ {\color{black}and $\delta\equiv0$}:
\[
\frac{d}{dt}V_{y}(x(t))\leq-(\varrho_{1}-\varrho_{2}\Delta^{2})
\]
for almost all $t\geq0$ while $V_{y}(x(t))\ne0$. 
\end{thm}
Thus, the control \eqref{eq:HOSM_ctrl} guarantees the global finite-time
stabilization of the origin independently on the disturbances $d\in\D$
in the absence of the measurement noise {\cb and mismatched terms}. Since the closed-loop system
\eqref{eq:LS}, \eqref{eq:HOSM_ctrl} is homogeneous of negative degree,
the ISS property with respect to {\cb $w,\delta\in\mathcal{L}_{\infty}^{n}$  with $\delta^\top b\equiv0$} follows
\cite{Bernuau2013}. The drawback of this regulator is the chattering
of the control signal appearing in a finite time once the trajectories
settled at the origin. {\color{black}
If the pair of matrices $(A, b)$ is controllable and the parameters $\varrho_1$ and $\varrho_2$ are properly chosen, then the LMI~(\ref{eq:LMI_imp_Rn}) is guaranteed to be feasible, see Proposition~12 of~\cite{Polyakov2016} for more details.
}

\subsection{Time-delay systems}

Consider a retarded functional differential equation \cite{kolmanovskii1986stability}:
\begin{equation}
\frac{dx(t)}{dt}=f(x_{t},d(t)),\quad t\geq0,\label{eq:systd}
\end{equation}
where $x(t)\in\mathbb{R}^{n}$ is the pointwise value of the state
vector $x_{t}\in C_{\eta}$, which is defined as $x_{t}(\theta)=x(t+\theta)$
for $-\eta\leq\theta\leq0$; $\eta>0$ is the delay; $d\in\L_{\infty}^{m}$
is the input, and with a slight abuse of notation we denote $\mathcal{D}=\{d\in\L_{\infty}^{m}:\|d\|_{\infty}<\Delta\}$,
$D=\{d\in\R^{m}:|d|\leq\Delta\}$ for a given $\Delta>0$. The functional
$f:C_{\eta}\times\mathbb{R}^{m}\to\mathbb{R}^{n}$ is upper semi-continuous
in the first argument, continuous in the second, and satisfies $f(0,0)=0$.
We consider system \eqref{eq:systd} with the initial condition $x_{0}\in C_{\eta}$.

It is known from theory of functional differential equations \cite{kolmanovskii1986stability}
that the system \eqref{eq:systd} with a locally Lipschitz $f$ has
a unique solution $x(t,x_{0},d)$ satisfying the initial condition
$x_{0}\in C_{\eta}$ for the input $d\in\L_{\infty}^{m}$, which is
defined on some finite time interval $[-\eta,T)$ (we will use the
notation $x(t)$ to reference $x(t,x_{0},d)$ if the origin of $x_{0}$
and $d$ is clear from the context). And the conditions of existence
of \emph{generalized} solutions for functional differential inclusions,
which can be obtained from convexification of \eqref{eq:systd}, can
be found in \cite{Kolmanovskii1999}. Further we will assume that
such solutions exist for \eqref{eq:systd}, and since in this work
the case of stability properties verified for all solutions issued
from a given initial condition is studied only (in the \emph{strong}
sense), the notation $x(t,x_{0},d)$ will be used to denote all solutions
from $x_{0}\in C_{\eta}$ with $d\in\D$.

\subsection{Stability definitions}

Let $\Omega$ be an open neighborhood of the zero function in $C_{\eta}$. 
\begin{defn}
\label{def:UGS} \cite{polyakov2015implicit,efimov2016delayed} The
trivial solution $x(t)=0$ of the system \eqref{eq:systd} is said
to be: 
\begin{itemize}
\item[(a)] \emph{uniformly Lyapunov stable} if there exists $\sigma\in\mathcal{K}$
such that for any $x_{0}\in\Omega$ and $d\in\D$, the solutions are
defined for all $t\geq0$ and satisfy $|x(t,x_{0},d)|\leq\sigma(\|x_{0}\|)$
for all $t\geq0$. 
\item[(b)] \emph{uniformly asymptotically stable} if it is uniformly Lyapunov
stable, \\ and $\lim_{t\to+\infty}|x(t,x_{0},d)|=0$ for any $x_{0}\in\Omega$
and $d\in\D$. 
\item[(c)] \emph{uniformly hyperexponentially} (or \emph{exponentially}) \emph{stable}
if it is uniformly Lyapunov stable, and there exist function $\Theta\in\mathcal{K}$
and a decay rate $\varkappa>0$, such that for all $t\geq0$ and any
$x_{0}\in\Omega$, $d\in\D$: $|x(t,x_{0},d)|\leq\Theta(\lVert x_{0}\rVert)e^{-e^{\varkappa t}}$
(or $|x(t,x_{0},d)|\leq\Theta(\lVert x_{0}\rVert)e^{-\varkappa t}$). 
\item[(d)] \emph{uniformly finite-time stable} if it is uniformly Lyapunov stable,
and for any $x_{0}\in\Omega$ and $d\in\D$ there exists $0\leq T<+\infty$
such that $x(t,x_{0},d)=0$ for all $t\geq T$. The functional $T_{0}(x_{0})=\sup_{d\in\D}\inf_{T\geq0}\{x(t,x_{0},d)=0\ \forall t\geq T\}$
is called the settling time of the system \eqref{eq:systd}. 
\end{itemize}
If $\Omega=C_{\eta}$, then the corresponding properties are termed
\emph{global uniform Lyapunov stability (GULS)}/\emph{asymptotic stability
(GUAS)}/\emph{hyperexponential stability (GUHeS)}/\emph{finite-time
stability (GUFTS)}. 
\end{defn}
For the forthcoming analysis, we will need the Lyapunov-Razumikhin
theorem, whose conventional formulation can be found in \cite{gu2003introduction,kolmanovskii1986stability},
and in a recent work \cite{nekhoroshikh2022hyperexponential} this
method was extended to hyperexponential stability (in that paper the
disturbance-free setting was considered, but all conditions and proofs
save their meaning for the uniform stability properties after a direct
adaptation of the formulations): 
\begin{thm}
\label{theo:LRE} \cite{nekhoroshikh2022hyperexponential} Let there
exist two locally Lipschitz continuous \\  Lyapunov-Razumikhin functions
$V_{1},V_{2}:\mathbb{R}^{n}\to\mathbb{R}_{+}$ such that: 
\begin{itemize}
\item[$\mathcal{R}_{e1}$)] For some $\alpha_{1,i},\alpha_{2,i}\in\mathcal{K}_{\infty}$ with
$i=1,2$ and for all $x\in\mathbb{R}^{n}$, we have 
\[
\alpha_{1,i}(|x|)\leq V_{i}(x)\leq\alpha_{2,i}(|x|).
\]
\item[$\mathcal{R}_{e2}$)] There exist constants $\beta_{1}>0$ and $\beta_{2}>0$ such that
\[
\beta_{1}|x|\leq\alpha_{1,1}(|x|)\quad\text{for all}\quad\beta_{1}|x|\leq1,
\]
\[
\beta_{2}|x|\leq\alpha_{1,2}(|x|)\quad\text{for all}\quad\beta_{2}|x|>1.
\]
\item[$\mathcal{R}_{e3}$)] For all $x\in\mathbb{R}^{n}$ such that $V_{2}(x)\leq1$, we have
\[
V_{1}(x)\leq1.
\]
\item[$\mathcal{R}_{e4}$)] For some $\chi>1$, $\gamma>0$ and for all $x_{t}\in C_{\eta}$
being a solution of \eqref{eq:systd} with $d\in\D$, denoting $V_{i}(t)=V_{i}(x(t))$,
$i=1,2$ we have: 
\begin{itemize}
\item[(a)] If $V_{2}(t)>1$ and $\underset{-\eta\leq\theta\leq0}{\max}V_{2}(t+\theta)\leq V_{2}(t)^{\chi}e^{\chi-1}$,
then 
\[
\dot{V}_{2}(t)\leq-\gamma\ln\left(eV_{2}(t)\right)V_{2}(t);
\]
\item[(b)] If $V_{1}(t)\leq1$ and $\underset{-\eta\leq\theta\leq0}{\max}V_{1}(t+\theta)^{\chi}\leq V_{1}(t)e^{\chi-1}$,
then 
\[
\dot{V}_{1}(t)\leq\gamma\ln\left(\frac{V_{1}(t)}{e}\right)V_{1}(t).
\]
\end{itemize}
\end{itemize}
Then, the system \eqref{eq:systd} is GUHeS at the origin with the
decay rate $\min\left\{ \frac{\ln\chi}{\eta},\gamma\right\} $. 
\end{thm}
An extension to the ILF application is also given in \cite{nekhoroshikh2022hyperexponential}
being slightly modified in \cite{Labbadi2024}: 
\begin{thm}
\label{theo:ILF} Let there exist two continuous functions $Q_{1},Q_{2}:\mathbb{R}_{+}^{\star}\times\mathbb{R}^{n}\setminus\{0\}\to\mathbb{R}$
such that for $i=1,2$: 
\begin{itemize}
\item[$\mathcal{R}_{i1}$)] $Q_{i}$ is a continuously differentiable function. 
\item[$\mathcal{R}_{i2}$)] for any $x\in\mathbb{R}^{n}\setminus\{0\}$ there exists $V_{i}\in\mathbb{R}_{+}^{\star}$
such that $Q_{i}(V_{i},x)=0$; 
\item[$\mathcal{R}_{i3}$)] there exist $\hbar_{1,i},\hbar_{2,i},\in\mathcal{IK}_{\infty}$ such
that 
\begin{align*}
\hbar_{1,i}(V_{i},|x|) & \leq Q_{i}(V_{i},x)\leq\hbar_{2,i}(V_{i},|x|)
\end{align*}
for all $V_{i}\in\mathbb{R}_{+}^{\star}$ and $x\in\mathbb{R}^{n}\setminus\{0\}$; 
\item[$\mathcal{R}_{i4}$)] there are constants $a>0$ and $b>0$ such that 
\[
\hbar_{1,1}(a|x|,|x|)\geq0\quad\text{for all}\quad a|x|\leq1
\]
and 
\[
\hbar_{1,2}(b|x|,|x|)\geq0\quad\text{for all}\quad b|x|>1;
\]
\item[$\mathcal{R}_{i5}$)] $\frac{\partial Q_{i}(V_{i},x)}{\partial V_{i}}<0$ for all $V_{i}\in\mathbb{R}_{+}^{\star}$
and $x\in\mathbb{R}^{n}\setminus\{0\}$; 
\item[$\mathcal{R}_{i6}$)] $Q_{1}(1,x)=Q_{2}(1,x)$ for all $x\in\mathbb{R}^{n}$; 
\item[$\mathcal{R}_{i7}$)] for some $\chi>1$, $\rho>0$, and for all $\varphi\in C_{\eta}$
and $d\in D$, we have: 
\begin{itemize}
\item[(a)] for $\Omega_{1}:=\{(s,\phi)\in\mathbb{R}_{+}^{\star}\times C_{\eta}:Q_{1}(s,\phi(0))=0,Q_{1}(1,\phi(0))\leq0, \\ \underset{-\eta\leq\theta\leq0}{\max}Q_{1}(s^{\frac{1}{\chi}}e^{1-\frac{1}{\chi}},\phi(\theta))\leq0\}$,
\begin{gather*}
(V_{1},\varphi)\in\Omega_{1}\Rightarrow\frac{\partial Q_{1}(V_{1},\varphi(0))}{\partial x}f(\varphi,d)\\
\leq-\rho\ln(V_{1}/e)V_{1}\frac{\partial Q_{1}(V_{1},\varphi(0))}{\partial V};
\end{gather*}
\item[(b)] for $\Omega{}_{2}:=\{(s,\phi)\in\mathbb{R}_{+}^{\star}\times C_{\eta}:Q_{2}(s,\phi(0))=0,Q_{2}(1,\phi(0))>0, \\ \underset{-\eta\leq\theta\leq0}{\max}Q_{2}(se^{\chi-1},\phi(\theta))\leq0\}$,
\begin{gather*}
(V_{2},\varphi)\in\Omega{}_{2}\Rightarrow\frac{\partial Q_{2}(V_{2},\varphi(0))}{\partial x}f(\varphi,d)\\
\leq\rho\frac{\partial Q_{2}(V_{2},\varphi(0))}{\partial V}.
\end{gather*}
\end{itemize}
\end{itemize}
Then, the system \eqref{eq:systd} is GUHeS at the origin. 
\end{thm}

\subsection{Robust stability}

A more detailed introduction to the ISS theory can be found in \cite{dashkovskiy2011input,sontag2007input}.
\begin{defn}
\label{def:ISStd} The system \eqref{eq:systd} is said to be ISS
if there exist $\beta\in\mathcal{KL}$ and $\iota\in\mathcal{K}$
such that 
\[
|x(t,x_{0},d)|\le\beta(\|x_{0}\|,t)+\iota(\|d\|_{\infty})
\]
for all $x_{0}\in C_{\eta}$, $d\in\mathcal{L}_{\infty}^{m}$ and
$t\geq0$. 
\end{defn}
If $d=0$, then we recover the conventional global asymptotic stability
 property \cite{dashkovskiy2011input}. 
\begin{defn}
\label{def:ISS_LRF} A locally Lipschitz continuous $V:\mathbb{R}^{n}\rightarrow\mathbb{R}_{+}$
is called an ISS Lyapunov-Razumikhin function (ISS-LRF) for \eqref{eq:systd}
if there exist $\sigma_{1},\sigma_{2}\in\mathcal{K}_{\infty}$ and
$\rho_{V},\rho_{d},\sigma\in\mathcal{K}$ such that the following
conditions are satisfied: 
\[
\sigma_{1}(|x|)\le V(x)\le\sigma_{2}(|x|)\quad\forall x\in\mathbb{R}^{n};
\]
\begin{gather}
V(\phi(0))\ge\max\left\{ \rho_{V}\left[\max_{\theta\in[-\eta,0]}V(\phi(\theta))\right],\rho_{d}(|d|)\right\} \label{eq:ISS_RL}\\
\Rightarrow D^{+}V(\phi(0))f(\phi,d)\le-\sigma(|\phi(0)|)\nonumber 
\end{gather}
for all $\phi\in C_{\eta}$ and $d\in\R^{m}$.
\end{defn}
\begin{thm}
\label{thm:Teel} \cite{teel1998connections} If there exists an ISS
Lyapunov-Razumikhin function $V$ for \eqref{eq:systd} and $\rho_{V}(s)<s$
for all $s>0$, then \eqref{eq:systd} is ISS.
\end{thm}
Under additional restrictions, existence of an ISS-LRF is also necessary
for stability \cite{efimov2024equivalence}.

\section{Design of {\color{black} delayed} HOSMC}

\label{sec:IV}

The idea of the method of \cite{Labbadi2024} consists in replacement
of the value of the ILF $V_{y}(y(t))$ derived through \eqref{eq:Q}
by its functional extension $\Psi(y_{t})$ given in \eqref{eq:Psi}
for a power $\chi>1$ and delay $\eta>0$. Then in the noise-free
case, applying Theorem \ref{theo:ILF} the following result can be
obtained (it is shown in the proof that $V_{y}(y(t))$ is an implicit
Lyapunov-Razumikhin function for the system):
\begin{thm}
\label{thm:He_Ctrl} \cite{Labbadi2024} Let the LMIs \eqref{eq:LMI_imp_Rn}
be feasible for some $X=X^{\top}\in\mathbb{R}^{n\times n}$, $Y\in\mathbb{R}^{1\times n}$
and $\varrho_{1}>\varrho_{2}\Delta^{2}>0$. If the control $u$ is
of the form 
\begin{equation}
u(y_{t})=YX^{-1}D_{\mathbf{r}}\left(\Psi^{-1}(y_{t})\right)y(t),\label{eq:u}
\end{equation}
where $\Psi(y_{t})$ is defined in \eqref{eq:Psi} with $\chi>1$
and $\eta>0$, $V_{y}\left(y(t)\right)$ is such that \\ $Q(V_{y}\left(y(t)\right),y(t))=0$
for $Q$ defined in \eqref{eq:Q}, then the system \eqref{eq:LS}
with $w\equiv0${\color{black}, $\delta \equiv 0$} and the control \eqref{eq:u} is GUHeS at the origin.
\end{thm}
Therefore, to apply \eqref{eq:u}, exactly the same conditions are
necessary to check as for \eqref{eq:HOSM_ctrl} having two auxiliary
free tuning parameters $\chi$ and $\eta$. Additional computational
power and memory are also required to calculate $\Psi$. As a result,
the matched disturbance can be compensated by a (theoretically) \emph{continuous}
control for all $t\in[0,+\infty)$ (the right end of the interval
can be included to highlight that it is necessary to wait the delay
time $\eta$ after the system is settled at the origin (for $t\to+\infty$)
and when we will get the discontinuity in the control \eqref{eq:u}).
The price for these advantages is a slightly slower convergence, since
\eqref{eq:HOSM_ctrl} guarantees for \eqref{eq:LS} with $w\equiv0$ {\color{black} and $\delta\equiv0$}
the GUFTS property (but still the decay with \eqref{eq:u} is faster than for any linear
control algorithm).

Note that in the presence of the noise $w\ne0$, the value $V_{y}\left(y(t)\right)$
is used in the control expression \eqref{eq:u}, while for the stability
analysis we need to consider the function defined by the equation $Q(V(x(t)),x(t))=0$ as an ISS-LRF candidate, whose connection with
$V_{y}(y(t))$ as a function of $w$ is not easy to evaluate.

\section{Main result}
\label{sec:V}

The main result of this note that establishes the conditions of robustness of \eqref{eq:u} in the noise $w$ {\cb and mismatched perturbation $\delta$} is given below:

\begin{thm}
\label{theo:u} Under conditions of Theorem \ref{thm:He_Ctrl}, the
system \eqref{eq:LS}, \eqref{eq:u} is ISS with respect to the input {\cb
$w,\delta\in\L_{\infty}^{n}$, $\delta^\top b\equiv0$} uniformly in $d\in\D$ provided that $\chi>1$
is chosen sufficiently close to $1$.
\end{thm}
\begin{proof}
Following Theorem \ref{thm:HOSM_ctrl}, the function $Q$ in \eqref{eq:Q}
verifies all conditions to define an ILF. Indeed, the function $Q$
is continuously differentiable for all $(V,x)\in\mathbb{R}_{+}^{\star}\times\mathbb{R}^{n}$,
and for any $x\in\R^{n}\setminus\{0\}$ there exists a solution $V(x)$
such that $Q(V(x),x)=0$ \cite{polyakov2015finite} (as usual we complement
this statement with $V(0)=0$). Moreover, the following chain of inequalities
can be established: 
\begin{gather*}
\sqrt{\lambda_{\min}(X)}\min\{V,V^{n}\}\leq|x|\leq\sqrt{\lambda_{\max}(X)}\max\{V,V^{n}\}
\end{gather*}
which are satisfied for all $(V,x)\in\mathbb{R}_{+}^{\star}\times\mathbb{R}^{n}$
with $Q(V,x)=0$, and $X$ is a positive definite matrix due to LMIs
\eqref{eq:LMI_imp_Rn}. We can express the partial derivatives as
follows: 
\begin{gather*}
\frac{\partial Q(V,x)}{\partial V}\\
=-2x^{\top}\left[\diag\{(n-i+1)V^{-n+i-2}\}_{i=1}^{n}\right]X^{-1}D_{\mathbf{r}}(V^{-1})x\\
=-V^{-1}x^{\top}D_{\mathbf{r}}(V^{-1})(X^{-1}G_{\mathbf{r}}+G_{\mathbf{r}}X^{-1})D_{\mathbf{r}}(V^{-1})x\\
<0
\end{gather*}
due to the imposed LMI $XG_{\mathbf{r}}+G_{\mathbf{r}}X\succ0$, and
\[
\frac{\partial Q(V,x)}{\partial x}=2x^{\top}D_{\mathbf{r}}(V^{-1})X^{-1}D_{\mathbf{r}}(V^{-1}).
\]
Recall the useful properties $D_{\mathbf{r}}(s)A=sAD_{\mathbf{r}}(s)$
and $D_{\mathbf{r}}(s)b=sb$ that are satisfied by construction.

Since all conditions of Theorem \ref{thm:He_Ctrl} are verified, for
$w\equiv0$ {\color{black}and $\delta\equiv0$} the closed-loop system \eqref{eq:LS}, \eqref{eq:u} is
GUHeS at the origin. To check the ISS property uniformly in $d\in\D$,
following Theorem \ref{thm:Teel}, we will use an ILF defined by the
equation $Q(V(x),x)=0$, and we need to examine the implication \eqref{eq:ISS_RL}
with the constraint $\rho_{V}(s)<s$ for all $s>0$, which has to
be satisfied for all $x_{t}\in\mathcal{C}_{\eta}$, {\cb $w,\delta\in\L_{\infty}^{n}$}
and $d\in\D$ {\cb using the restriction $\delta^\top b\equiv0$ (it is a necessary condition to have ISS property with respect to $\delta$ since the control is bounded)}. Define $V(t)=V(x(t))$ and note that
\[
\dot{V}(t)=-\left(\frac{\partial Q(V(t),x(t))}{\partial V}\right)^{-1}\frac{\partial Q(V(t),x(t))}{\partial x}\dot{x}(t),
\]
where the expressions of the partial derivatives are presented above,
$\dot{x}(t)$ comes from \eqref{eq:LS} after substitution of \eqref{eq:u}.
Consequently, skipping time dependence for brevity, and denoting $P=X^{-1}$,
$\mathbf{K}=YX^{-1}$, we obtain

\begin{gather*}
\frac{\partial Q(V,x)}{\partial x}\dot{x}=2x^{\top}D_{\mathbf{r}}(V^{-1})PD_{\mathbf{r}}(V^{-1})\left[Ax+b(u+d){\cb+\delta}\right]\\
=V^{-1}\bigg[x^{\top}D_{\mathbf{r}}(V^{-1})\left(A^{\top}P+PA+Pb\mathbf{K}\right.\\
\left.+\mathbf{K}^{\top}b^{\top}P\right)D_{\mathbf{r}}(V^{-1})x\\
+2x^{\top}D_{\mathbf{r}}(V^{-1})Pb\mathbf{K}\left(D_{\mathbf{r}}(\Psi^{-1})-D_{\mathbf{r}}(V^{-1})\right)x\\
+2x^{\top}D_{\mathbf{r}}(V^{-1})Pb\mathbf{K}D_{\mathbf{r}}(\Psi^{-1})w\\
+2x^{\top}D_{\mathbf{r}}(V^{-1})P(bd{\cb+VD_{\mathbf{r}}(V^{-1})\delta})\bigg].
\end{gather*}
After adding and subtracting the missing terms, we get for any $\gamma_{1},\gamma_{2},\gamma_{3}>0$:
\begin{gather*}
V\frac{\partial Q(V,x)}{\partial x}\dot{x}=\zeta^{\top}\mathcal{Q}\zeta\\
-\varrho_{1}x^{\top}D_{\mathbf{r}}(V^{-1})(PG_{\mathbf{r}}+G_{\mathbf{r}}P)D_{\mathbf{r}}(V^{-1})x+\varrho_{2}d^{2}\\
+\gamma_{1}w^{\top}D_{\mathbf{r}}(\Psi^{-1})^{2}w\\
+\gamma_{2}x^{\top}\left(D_{\mathbf{r}}(\Psi^{-1})-D_{\mathbf{r}}(V^{-1})\right)^{2}x{\cb+\gamma_{3}V^{2}\delta^\top D_{\mathbf{r}}^{2}(V^{-1})\delta},
\end{gather*}
where $\zeta=[\zeta_{1}\quad\zeta_{2}\quad\zeta_{3}\quad\zeta_{4}\quad\zeta_{5}]^\top$
with $\zeta_{1}=x^{\top}D_{\mathbf{r}}(V^{-1})$, $\zeta_{2}=d$,
$\zeta_{3}=w^{\top}D_{\mathbf{r}}(\Psi^{-1})$, $\zeta_{4}=x^{\top}\left(D_{\mathbf{r}}(\Psi^{-1})-D_{\mathbf{r}}(V^{-1})\right)${\cb, $\zeta_{5} = VD_{\mathbf{r}}(V^{-1})\delta$}
and 
\begin{gather*}
\mathcal{Q}=\begin{pmatrix}\Pi & Pb & Pb\mathbf{K} & Pb\mathbf{K} &{\cb  P}\\
b^{\top}P & -\varrho_{2} & 0 & 0 & {\cb 0}\\
(b\mathbf{K})^{\top}P & 0 & -\gamma_{1}I_{n} & 0 & {\cb 0}\\
(b\mathbf{K})^{\top}P & 0 & 0 & -\gamma_{2}I_{n} &{\cb 0}\\
{\cb P} & {\cb 0} & {\cb 0} & {\cb 0} & {\cb -\gamma_{3} I_n}
\end{pmatrix},\\
\Pi=A^{\top}P+PA+Pb\mathbf{K}+\mathbf{K}^{\top}b^{\top}P\\
+\varrho_{1}(PG_{\mathbf{r}}+G_{\mathbf{r}}P).
\end{gather*}
The LMIs \eqref{eq:LMI_imp_Rn} imply that there exist $\gamma_{1},\gamma_{2},\gamma_{3}>0$
such that $\mathcal{Q}\preceq0$. Indeed, multiplying the matrix $\mathcal{Q}$
by $\diag[X,1,X,X,X]$ from the left and right sides, we get the negative
semi-definite block $\left[\begin{array}{cc}
\Upsilon & b\\
b^{\top} & -\varrho_{2}
\end{array}\right]$ on the main diagonal, and the same property for the new matrix follows
for sufficiently big values of $\gamma_{1},\gamma_{2}$ and $\gamma_{3}$. Therefore,
one has 
\begin{gather*}
\frac{\partial Q(V,x)}{\partial x}\dot{x}\leq\frac{1}{V}\bigg[-\varrho_{1}x^{\top}D_{\mathbf{r}}(V^{-1})(PG_{\mathbf{r}}+G_{\mathbf{r}}P)D_{\mathbf{r}}(V^{-1})x\\
+\varrho_{2}\Delta^{2}+\gamma_{1}w^{\top}D_{\mathbf{r}}(\Psi^{-1})^{2}w\\
+\gamma_{2}x^{\top}\left(D_{\mathbf{r}}(\Psi^{-1})-D_{\mathbf{r}}(V^{-1})\right)^{2}x{\cb+\gamma_{3}V^{2}\delta^\top D_{\mathbf{r}}^{2}(V^{-1})\delta}\bigg]
\end{gather*}
and assume that the inequalities
\begin{gather}
w^{\top}D_{\mathbf{r}}\left(\Psi^{-1}\right)^{2}w\leq\frac{\varrho_{1}-\varrho_{2}\Delta^{2}}{6\gamma_{1}}x^{\top}D_{\mathbf{r}}\left(V^{-1}\right)(PG_{\mathbf{r}}\nonumber \\
+G_{\mathbf{r}}P)D_{\mathbf{r}}\left(V^{-1}\right)x,\label{eq:cond_A}\\
\left(D_{\mathbf{r}}\left(\Psi^{-1}\right)-D_{\mathbf{r}}\left(V^{-1}\right)\right)^{2}\leq\frac{\varrho_{1}-\varrho_{2}\Delta^{2}}{6\gamma_{2}}D_{\mathbf{r}}\left(V^{-1}\right)(PG_{\mathbf{r}}\nonumber \\
+G_{\mathbf{r}}P)D_{\mathbf{r}}\left(V^{-1}\right),\label{eq:cond_B}\\
{\color{black} V^{2}\delta^\top D_{\mathbf{r}}^{2}(V^{-1})\delta\leq\frac{\varrho_{1}-\varrho_{2}\Delta^{2}}{6\gamma_{3}}x^{\top}D_{\mathbf{r}}\left(V^{-1}\right)(PG_{\mathbf{r}}}\nonumber\\{\cb+G_{\mathbf{r}}P)D_{\mathbf{r}}\left(V^{-1}\right)x}\label{eq:cond_C}
\end{gather}
are verified (recall that $\varrho_{1}-\varrho_{2}\Delta^{2}>0$),
then
\begin{gather*}
\frac{\partial Q(V,x)}{\partial x}\dot{x}\leq-\frac{1}{V}\frac{\varrho_{1}-\varrho_{2}\Delta^{2}}{2}x^{\top}D_{\mathbf{r}}(V^{-1})(PG_{\mathbf{r}}\\
+G_{\mathbf{r}}P)D_{\mathbf{r}}(V^{-1})x,
\end{gather*}
which gives
\[
\dot{V}(t)\leq-\frac{\varrho_{1}-\varrho_{2}\Delta^{2}}{2},
\]
and it is left to show that the Lyapunov-Razumikhin conditions \eqref{eq:ISS_RL}
can provide \eqref{eq:cond_A}--\eqref{eq:cond_C} under a proper
choice of $\rho_{V}$ and $\rho_{d}$ {\cb for $w$ and $\delta$. Note that the inequalities \eqref{eq:cond_A} and \eqref{eq:cond_B}
depend on the noise $w$, the former explicitly, and the latter implicitly
since $\Psi$ is a function of $V_{y}$, while \eqref{eq:cond_C} under the constraint $\delta^\top b=0$ follows the inequality:
$$ \delta^\top \delta\leq\frac{\varrho_{1}-\varrho_{2}\Delta^{2}}{6\gamma_{3}\lambda_{\max}(P)}\min\{V^{2},V^{2(n-1)}\}$$
or, equivalently,
$$ V \geq \rho_{\delta}(|\delta|),$$
where
$$\rho_{\delta}(s)=\max\left[\sqrt{\frac{6\gamma_{3}\lambda_{\max}(P)}{\varrho_{1}-\varrho_{2}\Delta^{2}}s^{2}},\sqrt[2(n-1)]{\frac{6\gamma_{3}\lambda_{\max}(P)}{\varrho_{1}-\varrho_{2}\Delta^{2}}s^{2}}\right].$$

To derive the gain $\rho_d$ for the noise $w$}, observe that
\[
w^{\top}D_{\mathbf{r}}\left(\Psi^{-1}\right)^{2}w\leq\max\left(\Psi^{-2},\Psi^{-2n}\right)|w|^{2}
\]
and for $Q(V,x)=0$ there exist $\alpha>0$ such that
\[
\alpha\leq\frac{\varrho_{1}-\varrho_{2}\Delta^{2}}{4\gamma_{1}}x^{\top}D_{\mathbf{r}}\left(V^{-1}\right)(PG_{\mathbf{r}}+G_{\mathbf{r}}P)D_{\mathbf{r}}\left(V^{-1}\right)x,
\]
then \eqref{eq:cond_A} is implied by the following upper bound on
the noise amplitude:
\[
|w|\leq\sqrt{\alpha}\min\left(\Psi,\Psi^{n}\right).
\]
Next, the property \eqref{eq:cond_B} can be equivalently written
as 
\[
\left(D_{\mathbf{r}}\left(\frac{V}{\Psi}\right)-I_{n}\right)^{2}\leq\frac{\varrho_{1}-\varrho_{2}\Delta^{2}}{4\gamma_{2}}\left(PG_{\mathbf{r}}+G_{\mathbf{r}}P\right),
\]
and 
\[
\left(D_{\mathbf{r}}\left(\frac{V}{\Psi}\right)-I_{n}\right)^{2}\leq\max\Bigg[\left(\frac{V}{\Psi}-1\right)^{2},\left(\frac{V}{\Psi}-1\right)^{2n}\Bigg]I_n.
\]
Hence, the bound (we use here the fact that $PG_{\mathbf{r}}+G_{\mathbf{r}}P\succeq I_n$ due to LMIs \eqref{eq:LMI_imp_Rn})
\begin{gather*}
\frac{V}{\Psi}\leq\xi,\;
\xi=1+\min\Biggl\{\sqrt{\frac{\varrho_{1}-\varrho_{2}\Delta^{2}}{4\gamma_{2}}},
\sqrt[2n]{\frac{\varrho_{1}-\varrho_{2}\Delta^{2}}{4\gamma_{2}}}\Biggr\}
\end{gather*}
implies the condition \eqref{eq:cond_B}.

Let us assume that 
\begin{gather}
V_{y}\left(y(t)\right)\ge e^{1-\chi}\min\left[\underset{-\eta\leq\theta\leq0}{\max}V_{y}\left(y(t+\theta)\right),\right.\label{eq:V_y}\\
\left.\left(\underset{-\eta\leq\theta\leq0}{\max}V_{y}\left(y(t+\theta)\right)\right)^{\chi}\right],\nonumber 
\end{gather}
which ensures $\Psi(y_{t})=V_{y}(y(t)):=V_{y}(t)$, then the inequalities
\eqref{eq:cond_A} and \eqref{eq:cond_B} follow by
\begin{gather*}
|w(t)|\leq\sqrt{\alpha}\min\left(V_{y}(t),V_{y}^{n}(t)\right),\;V(t)\leq\xi V_{y}(t).
\end{gather*}
Assume that the function $\rho_{d}$ is chosen in a way that the relation
$V(t)\geq\rho_{d}(|w(t)|)$ enforces the following inequalities: 
\begin{gather}
(x+w)^{\top}D_{\mathbf{r}}\left(\xi V^{-1}\right)PD_{\mathbf{r}}\left(\xi V^{-1}\right)(x+w)\geq1\nonumber \\
\geq(x+w)^{\top}D_{\mathbf{r}}\left[(\xi V)^{-1}\right]PD_{\mathbf{r}}\left[(\xi V)^{-1}\right](x+w),\label{eq:rho_d}\\
|w|\leq\sqrt{\alpha}\min\left[\frac{V}{\xi},\left(\frac{V}{\xi}\right)^{n}\right].\nonumber 
\end{gather}
Since $(x+w)^{\top}D_{\mathbf{r}}\left(V_{y}^{-1}\right)PD_{\mathbf{r}}\left(V_{y}^{-1}\right)(x+w)=1$
by definition, and it was shown above that $\frac{\partial Q(V,x)}{\partial V}<0$,
the first two constraints on $w$ guarantee that 
\begin{equation}
\frac{V(t)}{\xi}\leq V_{y}(t)\leq\xi V(t),\label{eq:V_vs_Vy}
\end{equation}
hence, \eqref{eq:cond_B} is ensured. Moreover, in such a case the
last restriction on $w$ provides \eqref{eq:cond_A}. Finally, taking
\[
\rho_{V}(s)=e^{1-\chi}\xi^{\chi+1}\min\left[s,s^{\chi}\right],
\]
verification of the Lyapunov-Razumikhin conditions \eqref{eq:ISS_RL}
reads as follows:
\begin{gather*}
V(t)\ge\max\{\rho_{V}(\underset{-\eta\leq\theta\leq0}{\max}V(t+\theta)),\rho_{d}(|w(t)|),\rho_{\delta}(|\delta(t)|)\}\Rightarrow\\
\eqref{eq:V_vs_Vy},\;\eqref{eq:V_y}+\eqref{eq:cond_A}-\eqref{eq:cond_C}\Rightarrow\\
\dot{V}(t)\leq-\frac{\varrho_{1}-\varrho_{2}\Delta^{2}}{2},
\end{gather*}
which by Theorem \ref{thm:Teel} leads to the desired conclusion provided
that $e^{1-\chi}\xi^{\chi+1}<1$. Note that $\gamma_{2}$ can be chosen
arbitrary large, then $\xi$ can approach $1$ as close as we want
and there exists a value of $\chi>1$ that this condition is verified.
\end{proof}
According to the formulation of Theorem \ref{theo:u}, any Lyapunov
function matrix $P$ and the control gain $\mathbf{K}$ obtained in
theorems \ref{thm:HOSM_ctrl} or \ref{thm:He_Ctrl} can be used for
robust uniform hyperexponential stabilization (the same LMIs are utilized).
The proof of Theorem \ref{theo:u} suggests an expression for the
asymptotic gain of the system with respect to the noise $w$ in an implicit
form \eqref{eq:rho_d}{\cb, and for the mismatched perturbation $\delta$ the asymptotic gain $\rho_{\delta}$ is given in the closed form}.




{\cb The obtained robustness results with respect to measurement noise allow us to commutate easily to an output dynamic feedback. Indeed, assume that instead of the full state measurement in \eqref{eq:LS}, only the first component $x_1(t)\in\mathbb{R}$ of the state $x(t)$ is available:
$y_{1}(t)=x_{1}(t)+w_{1}(t), $
where $w_1\in\mathcal{L}^{1}_{\infty}$ is the respective noise, then due to canonical form of the system, a finite-time (or even fixed-time) convergent observer can be designed following the results from \cite{hosm:Levant:2003,fnt:observer:Andrieu:2008,fnt:observers:Spurgeon:2008,perruquetti2008finite,Lopez-Ramirez_etal2018:Aut,Ping2025}, then the output $y(t)$ in \eqref{eq:LS} becomes available being generated by an observer with the related estimation error $w(t)$ dependent on the used method and $d$, $\delta$ and $w_1$. In the disturbance/perturbation-free setting the estimation error $w$ converges to zero in a finite time, and the hyperexponential convergence rates for the state $x(t)$ can be recovered for the control \eqref{eq:u}. Since many of these observers guarantee ISS of the estimation error with respect to uncertainties, by the standard results, serial connection of ISS systems is ISS \cite{sontag2007input}, and such a dynamic output controller is robust. Consequently, the result of Theorem \ref{theo:u} is preserved in this case also.

Let us illustrate the obtained theoretical findings in numeric experiments.}

 \section{Simulations}\label{sec:VI} 
Consider the system described by \eqref{eq:LS} with \( n = 3 \), subject to a matched disturbance \( d(t) = \sin(10t) \), satisfying \( |d(t)| \leq 1 \), and a mismatched perturbation \( \delta(t) = \left[ 0.03\sin(3t)\;\; 0.05\cos(5t)\;\; 0 \right]^\top \). Additionally, the system is affected by noise terms \( w_1(t) = w_2(t) = w_3(t) = 0.1\rnd(1) \), where $\rnd(1)$ generates a uniformly distributed in the interval $[0, 1]$ random number.  The initial state vector is set as:
\(
x(0) = [ 0.1\; 1 \; 3 ]^{\top}.
\)
The closed-loop system is numerically simulated using the explicit Euler method with a fixed step size \( h = 5\times10^{-3} \), incorporating a time delay \( \eta = 0.1 \) and a parameter \( \chi = 1.1 \). To evaluate the performance of the proposed controller, we compare the obtained results with  HOSMC of order three described in \cite{polyakov2015finite,Polyakov2016}.   The control inputs are computed using the algorithm presented in \cite{polyakov2015finite,Polyakov2016}, with a minimum ILF value \( V_{\min} = 0.1 \), using homogeneous proportional control (see\footnote{Toolbox for MATLAB: \url{https://researchers.lille.inria.fr/~polyakov/hcs/tutorial.html}}
 for implementation details). The controller parameters are determined by solving the corresponding LMI system for \( \gamma_1 = 0.2 \), \( \gamma_2 = 2 \), and \( \varrho_1 = 1 \). The computed gain matrices are
\(
K = \begin{bmatrix} -310.4000 & -91.7333 & -12.0000 \end{bmatrix}\) and 
\(P = \begin{bmatrix}
60.2035 & 14.0373 & 1.1637 \\ 
14.0373 & 3.4227 & 0.3023 \\ 
1.1637 & 0.3023 & 0.0302 
\end{bmatrix}.\)
\begin{figure}[!ht]
\centering 
\includegraphics[scale=0.6]{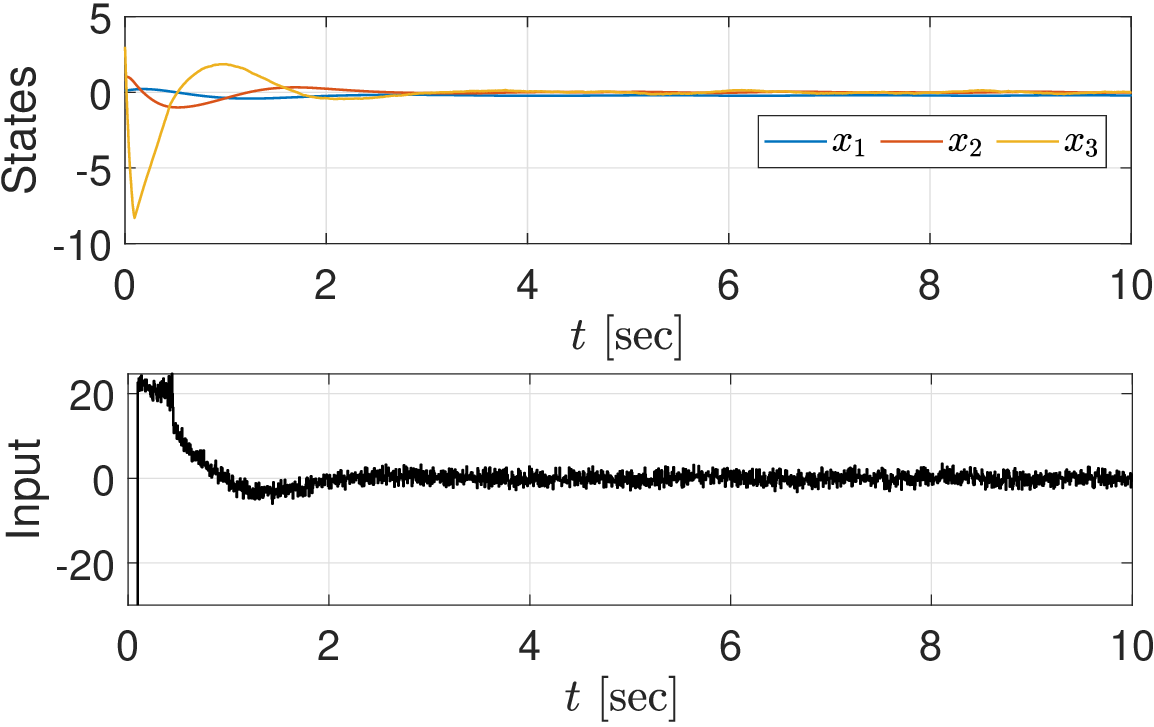} 
  \caption{Simulation results of the delayed HOSMC-ILF \eqref{eq:u}; Top: system states, Bottom: control input.}
\label{fig:Fig1} 
\end{figure}
\begin{figure}[!ht]
\centering 
\includegraphics[scale=0.6]{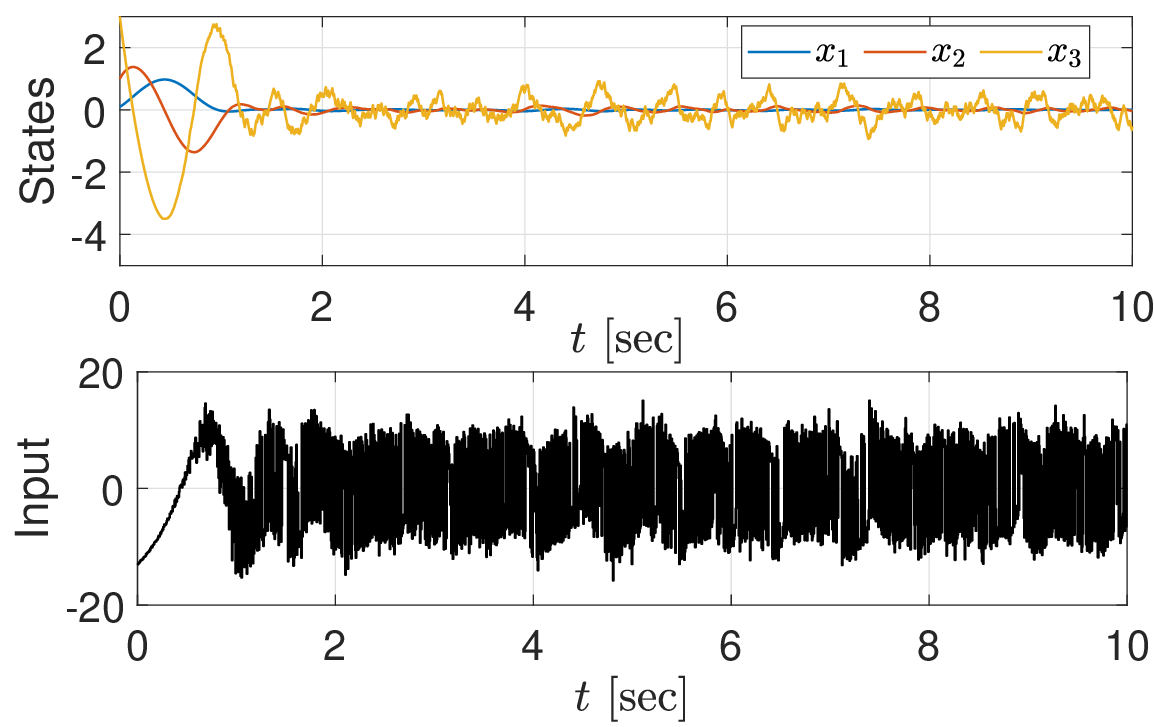} 
\caption{Simulation results of the  HOSMC-ILF (\ref{eq:HOSM_ctrl}) proposed in \cite{polyakov2015finite}; Top: system states, Bottom: control input.}
\label{fig:Fig2} 
\end{figure}
\begin{figure}[!ht]
\centering 
\includegraphics[scale=0.6]{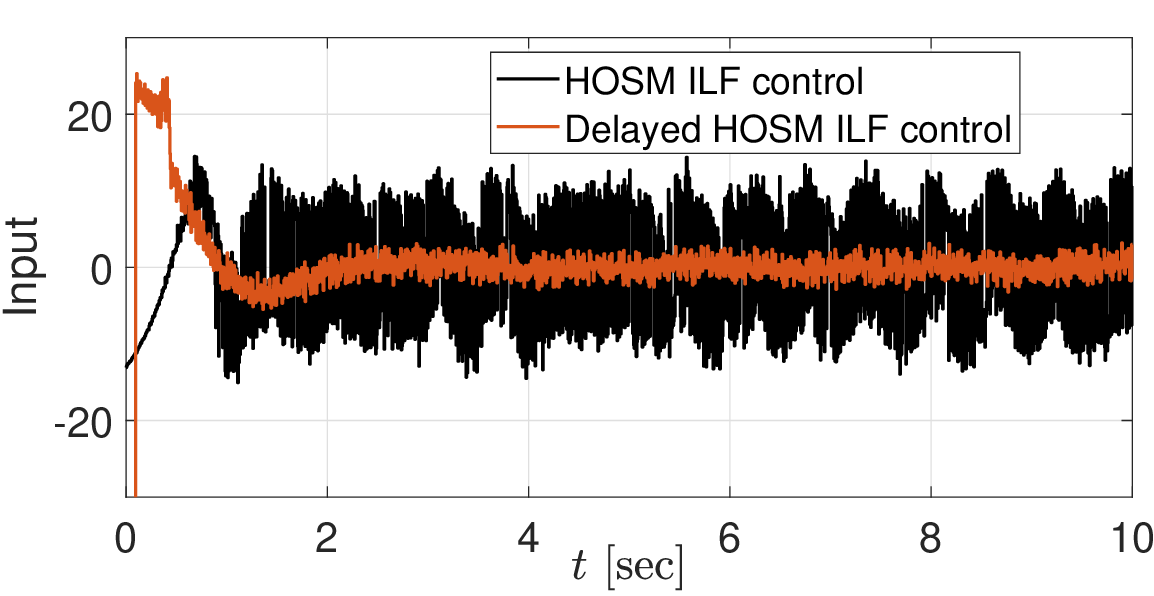} 
\caption{Comparison of control inputs.}
\label{fig:fig3} 
\end{figure}
The performance of the proposed delayed HOSM controller is illustrated in Fig. \ref{fig:Fig1}. It can be observed that  the control input maintains a reasonable amplitude.  
For comparison, the results of the HOSMC from \cite{polyakov2015finite} are depicted in Fig. \ref{fig:Fig2}. Although this approach ensures a faster tracking, it comes at the cost of chattering, as highlighted in Fig. \ref{fig:fig3}, which presents the control signals of both methods.  Note that this control was implemented using the sampled values of implicit Lyapunov function through the bisection algorithm given in  \cite{polyakov2015implicit, Labbadi2024}, and such a realization guarantees hyperexponential rate of convergence only.
Furthermore, the identification of the matched disturbance \( d(t) \) by \eqref{eq:u}, perturbed by the presence of $w$, can be observed  in Fig. \ref{fig:fig3}, demonstrating effective disturbance estimation after transients.

\section{Conclusions}

\label{sec:VIII} 
The robustness of the delayed HOSMC algorithm with respect to measurement noise {\cb and the mismatched properly structured disturbances} was established. Conventionally, an LMI-based tuning method for control parameters is used. Simulations show a significant reduction in chattering, affirming the effectiveness of the approach. 
Future work could explore extensions to observer design.

\bibliographystyle{siamplain}
\bibliography{biblios}
\end{document}